\documentclass{article}
\usepackage[utf8]{inputenc}
\usepackage{amsmath}
\usepackage{amsfonts}
\usepackage{amssymb}
\usepackage{amsthm}

\def\p{\partial}

\newtheorem{prop}{Proposition}
\theoremstyle{remark}
\newtheorem*{rem}{Remark}
\newtheorem{Rem}{Remark}
\newcommand{\dbar}{\bar{\partial}}
\newcommand{\wt}{\widetilde}
\newcommand{\be}{\begin{equation}}
\newcommand{\ee}{\end{equation}}
\newcommand{\bea}{\begin{eqnarray}}
\newcommand{\eea}{\end{eqnarray}}
\newcommand{\beaa}{\begin{eqnarray*}}
\newcommand{\eeaa}{\end{eqnarray*}}

\renewcommand{\d}{\mathrm{d}}

\newcommand{\pf}{\operatorname{pf}}
\usepackage{authblk}
\begin{document}
\author[1]{L.V. Bogdanov \thanks{leonid@itp.ac.ru}}
\author[2]{B.G. Konopelchenko}
\affil[1]{Landau ITP RAS, Moscow, Russia}
\affil[2]{
INFN, Sezione di Lecce,
73100 Lecce, Italy
}
\title{Integrability properties of symmetric 4+4-dimensional heavenly type equation}
\date{}
\maketitle
\begin{abstract}
We demonstrate that the dispersionless $\dbar$-dressing method developed before for general heavenly equation is applicable to the 4+4 and $2N+2N$ - dimensional 
symmetric heavenly type equations.
We introduce generating relation and 
derive the two-form defining the potential and equation for it.
We develop the dressing scheme, calculate a class of 
special solutions and demonstrate that reduction
from 4+4-dimensional  equation to four-dimensional
general heavenly equation can be effectively performed
on the level of the dressing data. 
We consider also the extension of the proposed scheme to
$2N+2N$-dimensional case.
\end{abstract}
\section{Introduction}
A simple and symmetric integrable 4+4-dimensional 
(TED) equation
\begin{multline}
(\Theta_{x^1y^2} - \Theta_{x^2y^1})(\Theta_{x^3y^4} - \Theta_{x^4y^3})
\\  
+ (\Theta_{x^2y^3} - \Theta_{x^3y^2})(\Theta_{x^1y^4} - \Theta_{x^4y^1})
\\
+ (\Theta_{x^3y^1} - \Theta_{x^1y^3})(\Theta_{x^2y^4} - \Theta_{x^4y^2}) = 0
\label{TED0}
\end{multline}
introduced in \cite{KS} represents a natural
generalisation of the four-dimensional general heavenly
equation \cite{Schief96}, \cite{DF10}, \cite{LVB15}.
Equation (\ref{TED0}) posesses interesting interpretations
in terms of differential geometry of K\"ahler
spaces \cite{KS}.
Lax pair for equation (\ref{TED0}) reads
\begin{gather}
\begin{split}
&(\Theta_{x^1y^2} - \Theta_{x^2y^1})D_3\Psi
+
(\Theta_{x^2y^3} - \Theta_{x^3y^2})D_1\Psi
+
(\Theta_{x^3y^1} - \Theta_{x^1y^3})D_2\Psi
=0,
\\
&(\Theta_{x^1y^2} - \Theta_{x^2y^1})D_4\Psi
+
(\Theta_{x^2y^4} - \Theta_{x^4y^2})D_1\Psi
+
(\Theta_{x^4y^1} - \Theta_{x^1y^4})D_1\Psi
=0,
\end{split}
\label{LaxTed0}
\end{gather}
where $D_i:=\p_{y^i}-\lambda\p_{x^i}$.
The existence of potential $\Theta$ providing
given representation
of coefficients of vector fields in the Lax pair  corresponds to vanishing divergence condition for
vector fields.

It is possible to consider
two-dimensional involutive distribution 
corresponding to
the Lax pair. A symmetric set of 
divergence free vector fields of
the form
\be
V_{ijk}=
(\Theta_{x^iy^j} - \Theta_{x^jy^i})D_k
+
(\Theta_{x^jy^k} - \Theta_{x^ky^j})D_i
+
(\Theta_{x^ky^i} - \Theta_{x^iy^k})D_j,
\label{symLax}
\ee
where $i,j,k$ is an arbitrary substitution of (distinct)
values 1,2,3,4, belongs to this distribution.

The goal of the present paper is to apply the 
techique of integrable dispersionless hierarchies
\cite{BK05}, \cite{BDM}, \cite{LVB09}, \cite{BK14}
to equation (\ref{TED0}), to introduce a dressing scheme
and construct a class of special solutions of equation
(\ref{TED0}) and its $2N+2N$ -dimensional generalisation.
The dressing scheme and generating equations developed 
in this paper are closely related to that of multidimensional
generalisation of six-dimensional heavenly equation
hierarchy with four degenerate wave functions
introduced in \cite{BP19}, and equation (\ref{TED0})
is obtained for a special choice of the set of times
(see also \cite{BP17}). The dressing scheme 
for equation (\ref{TED0}) corresponds
to the reduction of the dressing scheme for general
eight-dimensional integrable dispersionless hierarchy
\cite{LVB09}, and the functional freedom 
of the dressing data consists of functions of 
seven variables in accordance with eight-dimensional
integrability of equation (\ref{TED0}).

The paper is organized as follows. The applicability of 
techique of integrable dispersionless hierarchies
to equation (\ref{TED0}) is demonstrated in section  
\ref{SecInt}.
The dressing scheme is developed and exact solutions are constructed in section \ref{SecDress}. 
The $2N+2N$-dimensional TED equation and its solutions are considered in section \ref{SecN}.
\section{Integrability properties of equation (\ref{TED0})}
\label{SecInt}
\subsection*{Wave functions}
The structure of wave function 
is defined by 
linear problems (\ref{LaxTed0}).  
The set of wave functions contains four trivial wave functions
\bea
\phi^1=x^1+\lambda y^1, 
\quad \phi^2=x^2+\lambda y^2,
\quad \phi^3=x^3+\lambda y^3,
\quad \phi^4=x^4+\lambda y^4
\label{phi}
\eea
The presence of degenerate wave functions makes the
situation similar to six-dimensional heavenly equation,
where we have two functions exactly of type (\ref{phi})
\cite{BP19}.

To complete a basic set, which is
six-dimensional for integrable distribution 
with the basis (\ref{LaxTed0}),
we also introduce
two generic wave functions
\bea
\begin{aligned}
\Psi^1&=q + \wt\Psi^1,\quad
\wt\Psi^1=
\sum_{n=1}^\infty \Psi^1_n(p,q,\mathbf{x},\mathbf{y})
\lambda^{-n},
\quad u:=\Psi^1_1
\\
\Psi^2&=p + \wt\Psi^2,\quad
\wt\Psi^1=\sum_{n=1}^\infty \Psi^2_n(p,q,\mathbf{x},\mathbf{y})
\lambda^{-n},
\quad v:=\Psi^2_1
\end{aligned}
\label{Psi}
\eea
Till some moment we will consider
$p$, $q$ just as parameters (constants) not entering the equations, and later we will use them as variables in 
description of the
general framework of the hierarchy.
\subsection*{Generating relation}
The structure of wave functions for the Lax pair
(\ref{LaxTed0}) exactly corresponds to multidimensional
extension of six-dimensional heavenly equation
hierarchy
considered in \cite{BP19}, so we will briefly remind
some general results.
Integrable (involutive) distribution corresponding
to the set of wave functions (\ref{phi}),
(\ref{Psi}) (which represents a reduction of the wave
functions for the general hierarchy \cite{LVB09})
can be defined through the differential six-form
\bea
\Omega=\left((d\Psi ^{1}\wedge d\Psi ^{2})
\wedge (d\phi^{1}
\wedge d\phi
^{2}\wedge d\phi ^{4}
\wedge d\phi ^{4})\right),
\label{Omega}
\eea
where the differentials are taken with respect to
independent variables $\mathbf{x}$, $\mathbf{y}$,
imposing the generating relation
\bea
\Omega_-=\left((d\Psi ^{1}\wedge d\Psi ^{2})
\wedge (d\phi^{1}
\wedge d\phi
^{2}\wedge d\phi ^{4}
\wedge d\phi ^{4})\right)_-=0,
\label{Gen}
\eea
meaning that the projection of $\Omega$ to negative powers of $\lambda$
equals zero, thus $\Omega$ in our case is polynomial
(may be meromorphic in a more general setting).
The form (\ref{Omega}) satisfying relation (\ref{Gen})
defines an integrable distribution with
polynomial coefficients corresponding to
volume-preserving case (the basic vector fields cab
be chosen divergence-free).
The functions
$\phi$, $\Psi$ are wave functions for this distribution.

There are different ways to derive the basis of the 
distribution and compatibility conditions
(equations of the hierarchy) using generating equation
(\ref{Gen}). First we will give a direct derivation 
of the basic vectors
(\ref{symLax}) in the spirit of the dressing method.
Then we will modify the generating relation to provide
a simple and elegant way of direct derivation of 
equation (\ref{TED0}) using the language of differential tw0-forms. We will also introduce extra variables $p$, $q$ 
and give the interpretation of equation (\ref{TED0}) 
as a kind
of superposition principle for a set of
six-dimensional heavenly equations.
\subsection*{Lax pair -
direct derivation}
The set of wave functions defines a two-dimensional
involutive distribution annulating them, vector fields 
in this case can be taken in the form 
$$
{V}=\sum V_i D_i,\quad D_i:=\p_{y^i}-\lambda\p_{x^i},
$$
where $V_i$ are in general polynomial in $\lambda$.
Generating relation (\ref{Gen}) implies an
important property
\bea
\begin{pmatrix}
({V}\Psi^1)_+
\\
({V}\Psi^2)_+
\end{pmatrix}=0
\Rightarrow 
{V}
\begin{pmatrix}
\Psi^1
\\
\Psi^2
\end{pmatrix}=0,
\eea
allowing to construct polynomial vector fields
belonging to the integrable distribution explicitly,
eliminating `singular terms' in the result of the action
of vector fields on the wave functions.

Using this property and relations (compare (\ref{Psi}))
\beaa
(D_i \Psi^1)_+=-\p_{x^i} u, \quad 
(D_i \Psi^2)_+=-\p_{x^i} v,
\eeaa
we construct vector fields belonging to the distribution,
\begin{multline}
{V}_{ijk}=
((\p_{x^i} u)(\p_{x^j} v)- (\p_{x^j} u)(\p_{x^i} v))
D_k 
\\
+ 
((\p_{x^j} u)(\p_{x^k} v)- (\p_{x^k} u)(\p_{x^j} v))
D_i
\\
+
((\p_{x^k} u)(\p_{x^i} v)- (\p_{x^i} u)(\p_{x^k} v))
D_j  ,
\label{vf}
\end{multline}
where $i,j,k$ is an arbitrary substitution of
(distinct) values
$1,2,3,4$. 
The vanishing divergence
(or anti-self-adjointness) condition for vector fields
(\ref{vf})
implied by generating relation (\ref{Gen}),
\begin{multline*}
D_k((\p_{x^i} u)(\p_{x^j} v)- (\p_{x^j} u)(\p_{x^i} v))
\\
+
D_i((\p_{x^j} u)(\p_{x^k} v)- (\p_{x^k} u)(\p_{x^j} v))
\\
+
D_j((\p_{x^k} u)(\p_{x^i} v)- (\p_{x^i} u)(\p_{x^k} v)) 
=0,
\end{multline*}
leads to the existence of potential $\Theta$,
$$
((\p_{x^i} u)(\p_{x^j} v)- (\p_{x^j} u)(\p_{x^i} v))=
\Theta_{x^i y^j} - \Theta_{x^j y^i},
$$
and vector fields (\ref{vf}) take exactly the form
(\ref{symLax}), 
$$
V_{ijk}=
(\Theta_{x^iy^j} - \Theta_{x^jy^i})D_k
+
(\Theta_{x^jy^k} - \Theta_{x^ky^j})D_i
+
(\Theta_{x^ky^i} - \Theta_{x^iy^k})D_j,
$$
corresponding to two-dimensional integrable
distribution connected with equation (\ref{TED0}),
the basis (Lax pair) is given by an arbitrary pair of distinct
vector fields $V_{ijk}$.
\subsection*{Second form of generating relation}
Another form of generating relation, more suitable for 
vector fields of the form
$
{V}=\sum V_i D_i
$
is
\bea
(d^{w}\Psi^1\wedge d^{w}\Psi^2)_-=0,
\label{secondgen}
\eea
where we use the notations
\beaa
&&
w^i=y^i-\lambda^{-1}x^i,\quad
\wt w^i=y^i+\lambda^{-1}x^i=\lambda^{-1}\phi^i
\\
&&
\partial_{w^i}=\tfrac{1}{2}(\p_{y^i}-\lambda\p_{x^i})=
2D_i,
\quad
\partial_{\wt w^i}=\tfrac{1}{2}(\p_{y^i}+\lambda\p_{x^i}),
\\&&
d^{w}+d^{\wt w}= d,\quad 
d{\phi^i}\wedge d{w^i}=2 dx_i\wedge d y_i,
\eeaa
$d^{w}$ and $d^{\wt w}$ are differentials 
with respect to the subsets of variables $w_i$ and $\wt w_i$,
and the projection for 
the 2-form
\bea
\omega:=d^{w}\Psi^1\wedge d^{w}\Psi^2
=\sum_{n} \omega_{ij}^{(n)}(\mathbf{x},\mathbf{y})
\lambda^n 
d w^i \wedge d w^j
\label{omega}
\eea
is understood as
\bea
\omega_-=\sum_{n<0} \omega_{ij}^{(n)}(\mathbf{x},\mathbf{y})
\lambda^n
dw^i\wedge dw^j,
\label{projection2}
\eea
where the coefficients $\omega_{ij}^{(n)}$ do not depend on
$\lambda$.
\begin{prop}
Generating relation (\ref{Gen}) 
\beaa
\left((d\Psi ^{1}\wedge d\Psi ^{2})
\wedge (d\phi^{1}
\wedge d\phi
^{2}\wedge d\phi ^{4}
\wedge d\phi ^{4})\right) _{-}=0,
\eeaa
considered for the set of independent variables
$x_i$, $y_i$, is equivalent to the relation
(\ref{secondgen})
\beaa
(d^{w}\Psi^1\wedge d^{w}\Psi^2)_-=0,
\eeaa
where the projection in the second relation is defined
by  (\ref{projection2}).
\end{prop}
\begin{proof}
First, using the relations $d=d^{w}+d^{\wt w}$ and
$\wt w^i=\lambda^{-1}\phi^i$,  we get 
\beaa
&&
d\Psi^{1}\wedge d\Psi^{2}
\wedge (d\phi^{1}
\wedge d\phi
^{2}\wedge d\phi ^{4}
\wedge d\phi ^{4})
\\&&\qquad
=
d^{w}\Psi^1\wedge d^{w}\Psi^2\wedge (d\phi^{1}
\wedge d\phi
^{2}\wedge d\phi ^{4}
\wedge d\phi ^{4}).
\eeaa
Having in mind expression (\ref{omega}) and taking
into account that
\beaa
dw^i\wedge\d\phi^i=2 d y^i\wedge d x^i,
\eeaa
it is easy to check necessary and sufficient conditions 
of the Proposition.
\end{proof}

Then we have
\beaa
\omega=d^{w}\Psi^1\wedge d^{w}\Psi^2
=
(d^{w}\Psi^1\wedge d^{w}\Psi^2)_+=
\omega_{ij}d w^i\wedge d w^j,
\eeaa
where $\omega_{ij}$ is independent of $\lambda$.
Evidently,
\bea
\omega\wedge\omega=0,\quad d^{w}\omega=0,
\label{2formprop}
\eea
and for $\omega_{ij}$ we get the relations obtained
in \cite{KS}.

Second relation in (\ref{2formprop}) implies
the existence of potential (see \cite{KS}
for more detail)
\beaa
\omega_{ij}=\Theta_{x^i y^j} - \Theta_{x^j y^i},
\eeaa
and first relation (\ref{2formprop}) gives equation
(\ref{TED0})
for potential $\Theta$.
\subsection*{Connection with six-dimensional
heavenly equation. The hierarchy framework}
Considering more general wave functions $\phi^i$
containing higher times
\bea
\phi^i=x^i+\lambda y^i+
\sum_{n=2}^N t^i_n\lambda^n
\label{phiH}
\eea
and treating $q$, $p$ in wave functions (\ref{Psi})
as independent variables, we put ourselves to the 
context of multidimensional extension of
the six-dimensional heavenly equation hierarchy
considered in \cite{BP19}
(which is in turn a reduction of general dispersionless
hierarchy \cite{LVB09}) . Generating relation 
(\ref{Gen}) preserves its form for functions 
(\ref{phiH}) and extended set of independent variables,
\beaa
\Omega_-=\left((d\Psi ^{1}\wedge d\Psi ^{2})
\wedge (d\phi^{1}
\wedge d\phi
^{2}\wedge d\phi ^{4}
\wedge d\phi ^{4})\right)_-=0,
\eeaa
First, this relation implies that
$u=-\Theta_p$, $v=\Theta_q$,
and potential $\Theta$ satisfies six-dimensional
heavenly equations
\begin{equation}
\Theta_{x^i y^j} - \Theta_{x^j y^i} = \{ \Theta_{x^i}, \Theta_{x^j}\}_{(q,p)} , 
\label{6DH}
\end{equation}
with the Lax pair
\bea
\begin{aligned}
D_i \Psi +\{\Theta_{x^i},\Psi\}_{(q,p)}&=0,
\\
D_j \Psi +\{\Theta_{x^j},\Psi\}_{(q,p)}&=0,
\end{aligned}
\label{6DHLax}
\eea
for any pair of distinct $i,j\in 1,2,3,4$.
Here we use the Poisson bracket          
$$
\{ f_1, f_2\}_{(q,p)} := \frac{\partial f_1}{\partial q} 
\frac{\partial f_2}{\partial p}  -  \frac{\partial f_1}{\partial p} 
\frac{\partial f_2}{\partial q}. 
$$
Taking into account that $u=-\Theta_p$, $v=\Theta_q$ 
and using six-dimensional heavenly equations
we get
\beaa
((\p_{x^i} u)(\p_{x^j} v)- (\p_{x^j} u)(\p_{x^i} v))=
\{ \Theta_{x^i}, \Theta_{x^j}\}_{(q,p)}=
\Theta_{x^i y^j} - \Theta_{x^j y^i},
\eeaa
thus providing
another proof that vector fields (\ref{vf}) take the form corresponding
to the Lax pair (\ref{LaxTed0}). Vector fields
(\ref{vf}) can be obtained as combinations
of three vector fields of the form (\ref{6DHLax})
with eliminated Hamiltonian part. And equation
(\ref{TED0}) may be considered as a kind
of superposition principle (intertwining equation)
for a set of consistent six-dimensional heavenly equations
(\ref{6DH}) with different $i,j\in 1,2,3,4$.
\section{Dressing scheme and special solutions}
\label{SecDress}
Dressing scheme to construct solutions of generating
relations of the type (\ref{Gen}) was inroduced
in \cite{BP19} as a reduction of a dressing scheme
for general multidimensional dispersionless
hierarchies \cite{LVB09}.
It is formulated in terms of Riemann-Hilbert
problem on the unit circle (or the boundary of some
region $G$)
\beaa
\Psi ^{1}_\text{in} =
F^1(\lambda,\Psi ^{1},\Psi^2;
\phi^{1},\phi^{2},\phi^{3},\phi^{4})_\text{out},
\\
\Psi ^{2}_\text{in} =
F^2(\lambda,\Psi ^{1},\Psi^2;
\phi^{1},\phi^{2},\phi^{3},\phi^{4})_\text{out},
\eeaa
where the diffeomorphism defined by $F_1$, $F_2$
for the case of Hamiltonian reduction should be
area-preserving with respect to the variables
$\Psi^1$, $\Psi^2$.
Alternatively, it is possible to use the 
$\dbar$ problem in the unit disk (or
some region $G$)
\bea
\begin{aligned}
\dbar\Psi ^{1}& =
W_{,2}(\lambda,\bar\lambda,\Psi ^{1},\Psi^2;
\phi^{1},\phi^{2},\phi^{3},\phi^{4}),\quad &W_{,2}&:=\frac{\p W}{\p \Psi^2},
\\
\dbar \Psi ^{2}& =
-W_{,1}(\lambda,\bar\lambda, \Psi ^{1},\Psi^2;
\phi^{1},\phi^{2},\phi^{3},\phi^{4}),\quad
&W_{,1}&:=\frac{\p W}{\p \Psi^1}.
\end{aligned}
\label{dbar}
\eea
Here the Hamiltonian reduction is taken into account 
explicitly. We search for the solutions of the form
\beaa
\Psi^1=q+ \tilde \Psi^1,\quad \Psi^2=p+ \tilde \Psi^2
\eeaa
where $\tilde \Psi^1$, $\tilde \Psi^2$ are analytic outside $G$
and go to zero at infinity. The series for these
functions at infinity give a solution to generating
relation (\ref{Gen}) at infinity. The form $\Omega$
defined through these functions due to the problem (\ref{dbar})
is analytic in
the complex plane (may be meromorphic in more general 
setting).

The functional freedom of
the dressing data consists of functions of 7 variables,
that corresponds to the functional freedom
for general solution of  equation (\ref{TED0}),
indicating eight-dimensional integrability of equation
(\ref{TED0}).
Below we will consruct a class of solutions for 
equation (\ref{TED0}). 
We will
go along the lines of similar calculations for
general heavenly equation 
presented in \cite{LVB15} and 
for six-dimensional heavenly equation in \cite{BP19}.
\subsection*{Special solutions}
A class of solutions for  equation 
(\ref{TED0})
in terms of implicit functions (similar 
to solutions of hyper-K\"ahler hierarchy
\cite{Gindikin86},
\cite{Takasaki89})
can be constructed using the choice
\beaa
&&
\frac{1}{2\pi\mathrm{i}}W(\lambda,\bar \lambda,\Psi^1,\Psi^2;\phi^1,\phi^2,\phi^3,\phi^4)=
\\
&&\quad
=\sum_{i=1}^{M}\delta(\lambda-\mu_i)
F_i(\Psi^1;\phi^1,\phi^2,\phi^3,\phi^4)
+
\sum_{i=1}^{M}\delta(\lambda-\nu_i)
G_i(\Psi^2;\phi^1,\phi^2,\phi^3,\phi^4),
\eeaa
where $\delta(\lambda-\mu_i)$, $\delta(\lambda-\nu_i)$ are two-dimensional
delta functions in the complex plane 
characterized
by the relation $\dbar \lambda^{-1}=2\pi\mathrm{i}\delta(\lambda)$
(the constant defines normalization of two-dimensional
delta function),
and $F_i$, $G_i$ are 
arbitrary
(complex-analytic) functions of
three variables. The $\dbar$ problem (\ref{dbar}) in this case reads
\bea
\begin{aligned}
\dbar \Psi^1&=2\pi\mathrm{i}\sum_{i=1}^{M}
\delta(\lambda-\nu_i)G'_i(\Psi^2;\phi^1,\phi^2,\phi^3,\phi^4),
\quad &G'_i&=\frac{\p G_i}{\p \Psi^2},
\\
\dbar \Psi^2&=-2\pi\mathrm{i}\sum_{i=1}^{M}
\delta(\lambda-\mu_i)F'_i(\Psi^1;\phi^1,\phi^2,\phi^3,\phi^4),
\quad &F'_i&=\frac{\p F_i}{\p \Psi^1},
\end{aligned}
\label{dbardelta}
\eea
where, due to delta functions, $G'_i$, $F'_i$
are taken at $\lambda=\nu_i$, $\lambda=\mu_i$, 
and inverting operator $\dbar$, we get
\bea
\begin{aligned}
\Psi^1-q&=\sum_{i=1}^{M}
(\lambda-\nu_i)^{-1}
G'_i(\Psi^2(\nu_i);\phi^1(\nu_i),\phi^2(\nu_i),\phi^3(\nu_i),\phi^4(\nu_i)),
\\
\Psi^2-p&=-\sum_{i=1}^{M}
(\lambda-\mu_i)^{-1}
F'_i(\Psi^1(\mu_i);\phi^1(\mu_i),\phi^2(\mu_i),\phi^3(\mu_i),\phi^4(\mu_i)).
\end{aligned}
\label{dbardelta1}
\eea
The solutions of the $\dbar$ problem are then of
the form
\bea
\Psi^1=q+\sum_{i=1}^{M} \frac{f_i}{\lambda-\nu_i}, \quad  
\Psi^2=p+\sum_{i=1}^{M} \frac{g_i}{\lambda-\mu_i},
\label{psi}
\eea
and from (\ref{dbardelta1}) the functions $f_i$, $g_i$ are defined as implicit functions,
\be
\begin{aligned}
f_i(\mathbf{x},\mathbf{y})&=G'_i\left(p+
\sum_{k=1}^{M} \frac{g_k(\mathbf{x},\mathbf{y})}{\nu_i-\mu_k};
\phi^1(\nu_i),\phi^2(\nu_i),\phi^3(\nu_i),\phi^4(\nu_i)
\right),
\\
g_i(\mathbf{x},\mathbf{y})&=-F'_i\left(q+
\sum _{k=1}^{M}\frac{f_k(\mathbf{x},\mathbf{y})}{\mu_i-\nu_k};
\phi^1(\mu_i),\phi^2(\mu_i),\phi^3(\mu_i),\phi^4(\mu_i)
\right),
\end{aligned}
\label{impl}
\ee 
where equations (\ref{impl}) represent a closed
system of $2M$ equations for $2M$ functions
$f_i$, $g_i$, defining them as functions
of $\mathbf{x}$, $\mathbf{y}$, $p$, $q$.
The potential $\Theta$ 
is then given by the general formula (see \cite{BK05})
\bea
\Theta(\mathbf{x},\mathbf{y})=
\frac{1}{2\pi\mathrm{i}}\iint_{G}
\Bigl(\wt \Psi^2 \dbar \wt \Psi^1
-
W(\lambda,\bar \lambda,\Psi^1,\Psi^2;\phi^{1},\phi^{2},\phi^{3},\phi^{4})
\Bigr) 
d\lambda\wedge d\bar \lambda,
\label{HEtau}
\eea
it depends on the set
of arbitrary functions of {\em five variables} $F_i$, $G_i$,
\begin{multline} 
\Theta(\mathbf{x},\mathbf{y})=
\sum_{i=1}^{M} 
F_i(\Psi^1;\phi^1,\phi^2,\phi^3,\phi^4)
\bigr|_{\lambda=\mu_i}
+
\sum_{i=1}^{M} 
G_i(\Psi^2;\phi^1,\phi^2,\phi^3,\phi^4)
\bigr|_{\lambda=\nu_i}
\\
+\sum_{i=1}^{M}\sum_{j=1}^{M} \frac{f_i g_j}{\nu_i - \mu_j},
\label{Theta}
\end{multline} 
where $\Psi^1$, $\Psi^2$ are given by
(\ref{psi}), $\phi^1$, $\phi^2$, $\phi^3$, $\phi^4$  
are of the form
(\ref{phi})
and the functions $f_i$, $g_i$ are defined as implicit functions 
by equations (\ref{impl}). Formula (\ref{Theta}) corresponds
to the special solution of hyper-K\"ahler hierarchies 
presented
in \cite{Takasaki89}, however, it is
important to note that in our case the solution depends
on the set of arbitrary functions of {\em five variables},
in contrast to the set of functions of one variable
in \cite{Takasaki89} and functions of
three variables in \cite{BP19}.
\begin{rem}
It is possible to prove directly that ansatz
(\ref{psi}), taking into account relations (\ref{impl}),
gives a solution to generating equation (\ref{Gen}).
\end{rem}
\subsection*{From equation (\ref{TED0}) to the 
general heavenly equation}
The connection between equation (\ref{TED0})
and four-dimensional general heavenly equation
is described in \cite{KS} in a simple and elegant
way, namely as a travelling way reduction $\p_{y^i}=
\lambda_i\p_{x_i}$. It is remarkable that this reduction
can be rather easily performed in terms of the dressing
data and, specifically, for the special solution of
the type (\ref{Theta}), after some minor modification.
This reduction gives a general way of introducing
vertex variable (corresponding to simple pole
at some point) instead of a pair of variables of 
$x,y$, thus providing different types of
generating equations for lower-dimensional systems.

First, we slightly modify the definition of the wave functions (\ref{Psi}),
\bea
\begin{aligned}
\Psi^1&=q + 
\sum_1^4 a_i y_i + \wt\Psi^1,
\\
\Psi^2&=p + 
\sum_1^4 b_i y_i + \wt\Psi^2
\end{aligned}
\label{Psimod}
\eea
leaving the functions $\phi^i$ intact. It is easy to check that this modification leads to equation (\ref{TED0}) 
and six-dimensional heavenly equations
for the potential $\Theta$ containing vacuum terms,
\beaa
\Theta=\Theta_0+ \wt\Theta,
\quad \Theta_0=\tfrac{1}{2}\sum(a_ib_j-a_jb_i)x^iy^j +
\sum (a_i p x^i - b_i q x^i).
\eeaa 
To perform a transition from 8-dimensional to 4-dimensional
case, let us consider $\dbar$-data of the form
\bea
W(\lambda,\bar\lambda,\Psi ^{1},\Psi^2;
\phi^{1},\phi^{2},\phi^{3},\phi^{4})=
F\left(\lambda,\bar\lambda,
\Psi ^{1}+\sum \frac{a_i\phi^i}{\lambda-\lambda_i},
\Psi ^{2}+\sum \frac{b_i\phi^i}{\lambda-\lambda_i}\right),
\label{dwave}
\eea
where $F$ vanishes in the neighborhoods of infinity
and points $\lambda=\lambda_i$.
First, these data evidently correspond to some solution of
equation (\ref{TED0}) with the vacuum background, and, taking into
acount that
\begin{multline}
F\left(\lambda,\bar\lambda,
\Psi ^{1}+\sum \frac{a_i\phi^i}{\lambda-\lambda_i},
\Psi ^{2}+\sum \frac{b_i\phi^i}{\lambda-\lambda_i}\right)
\\
=
F\left(\lambda,\bar\lambda,
\wt\Psi ^{1}+
\sum \frac{a_i(x^i-\lambda_i y^i)}{\lambda-\lambda_i},
\wt\Psi ^{2}+
\sum \frac{b_i(x^i-\lambda_i y^i)}{\lambda-\lambda_i}\right)
\label{F4}
\end{multline}
the solution is of the form
\bea
\Theta=\sum(a_ib_j-a_jb_i)x^iy^j +
\wt\Theta(x^1-\lambda_1 y^1,x^2-\lambda_2 y^2,
x^3-\lambda_3 y^3, x^4-\lambda_4 y^4),
\label{TEDrun}
\eea
where $\wt\Theta$ corresponds to the travelling waves reduction. Then four-dimensional potential
(with the same $\wt\Theta$)
\bea
\Theta=
\tfrac{1}{2}\sum_{i \neq j}\frac{(a_ib_j-a_jb_i)x^ix^j}{\lambda_i-\lambda_j} +
\wt\Theta(x^1,x^2,
x^3, x^4),
\label{Genrun}
\eea
satisfies general heavenly equation
\begin{multline}
(\lambda_2-\lambda_3)(\lambda_1-\lambda_4)\Theta_{x^2 x^3}\Theta_{x^1 x^4}
\\
+(\lambda_1-\lambda_3)(\lambda_4-\lambda_2)
\Theta_{x^1 x^3}\Theta_{x^4 x^2}
\\
+
(\lambda_1-\lambda_2)(\lambda_3-\lambda_4)
\Theta_{x^1 x^2}\Theta_{x^3 x^4}=0.
\label{GH}
\end{multline}
\begin{Rem}
Dressing data of the form (\ref{F4})
(in terms of wave variables)
exactly correspond
to the dressing scheme developed in \cite{LVB15}
for general heavenly equation. And in the process
of reduction from equation (\ref{TED0}) to 
general heavenly equation (\ref{GH}) we obtain the same vacuum term that was obtained in \cite{LVB15} directly
from the dressing scheme for general heavenly equation.
\end{Rem}
To get a solution 
of the form (\ref{Theta}) corresponding
to the travelling wave reduction, we should use
modified wave functions (\ref{Psimod})
and the dressing data of the form
\begin{multline}
\frac{1}{2\pi\mathrm{i}}W(\lambda,\bar \lambda,\Psi^1,\Psi^2;\phi^1,\phi^2,\phi^3,\phi^4)=
\\
=\sum_{k=1}^{M}\delta(\lambda-\mu_k)
F_k\bigl(\Psi ^{1}+\sum \tfrac{a_i\phi^i}{\lambda-\lambda_i}\bigr)
+
\sum_{k=1}^{M}\delta(\lambda-\nu_k)
G_k\bigl(\Psi ^{2}+\sum \tfrac{b_i\phi^i}{\lambda-\lambda_i}\bigr),
\label{swave}
\end{multline}
where $F_k$, $G_k$ are now (complex-analytic) functions
of one variable. These dressing data correspond 
to solution of equation (\ref{TED0}) of
the form (\ref{TEDrun}), where the potential
$\wt\Theta$ is given by fomulae (\ref{HEtau}),
(\ref{Theta}) (taking into account
minor modifications of 
wave functions (\ref{psi}) and equations (\ref{impl})
corresponding to expressions (\ref{Psimod})).

\section{$2N+2N$-dimensional generalisation}
\label{SecN}
Let us briefly discuss how to generalise
the presented results to multidimensions.
It was demonstrated in \cite{KS} that, similar to
the case of general heavenly equation \cite{LVB15},
equation (\ref{TED0}) can be generalized to multidimensional case. 
This $2N+2N$-dimensional equation has the form \cite{KS}
\be
\epsilon^{i_1\cdots i_{2N}}\omega_{i_1i_2}\cdots\omega_{i_{2N-1}i_{2N}} = 0
\label{TEDN}
\ee
or $\pf(\omega)=0$, where $\pf(\omega)$ is
a Pfaffian and $\omega$ is
$2N$$\times$$2N$ skew-symmetric matrix 
with the entries
\bea
\omega_{ij}=\Theta_{x^i y^j} - \Theta_{x^j y^i}.
\label{potN}
\eea
Another elegant
way to describe equation (\ref{TEDN})
it is to use the differential two-form $\omega$,
$$
\omega=\omega_{ij}(\mathbf{x},\mathbf{y})dw^i\wedge dw^j,
$$
where coefficients
$\omega_{ij}$ are independent of $\lambda$.
For 4+4-dimensional case equation (\ref{TED0}) is 
equivalent to conditions (\ref{2formprop})
for the form $\omega$. Equation (\ref{TEDN})
is equivalent to the conditions
\bea
\wedge^N\omega=0,\quad d^{w}\omega=0,
\label{2formpropN}
\eea
where second condition, independent of dimensionality,
implies the existence of potential $\Theta$ (\ref{potN}),
and first condition gives equation (\ref{TEDN})
for the potential generalising equation
(\ref{TED0}) to multidimensions (see \cite{KS} for
more detail).

The analogue of generating relation (\ref{Gen})
for this case is
\begin{multline}
\Omega_-=\bigl((d\Psi ^{1}\wedge d\Psi ^{2}+
d\Psi ^{3}\wedge d\Psi ^{4}
+\dots+d\Psi ^{2N-3}\wedge d\Psi ^{2N-2})
\\
\wedge (d\phi^{1}
\wedge d\phi
^{2}\wedge \dots\wedge
d\phi ^{2N})\bigr)_-=0,
\label{GenN}
\end{multline}
where functions $\phi^i$ are of the form 
(\ref{phi}) and 
the series the functions 
$\Psi^i$ are of the type (\ref{Psi}),
\beaa
\begin{aligned}
\Psi^{2k-1}&=q^k + \wt\Psi^{2k-1},\quad
&\wt\Psi^{2k-1}&=
\sum_{n=1}^\infty \Psi^{2k-1}_n(p,q,\mathbf{x},\mathbf{y})
\lambda^{-n},
\\
\Psi^{2k}&=p^k + \wt\Psi^{2k},\quad
&\wt\Psi^{2k}&=\sum_{n=1}^\infty \Psi^{2k}_n(p,q,\mathbf{x},\mathbf{y})
\lambda^{-n},\quad 1\leqslant k\leqslant N-1.
\end{aligned}
\eeaa
The analogue of two-form $\omega$ (\ref{omega})
is
\beaa
\omega=d^w\Psi ^{1}\wedge d^w\Psi ^{2}+
d^w\Psi ^{3}\wedge d^w\Psi ^{4}
+\dots+d^w\Psi ^{2N-3}\wedge d^w\Psi ^{2N-2},
\eeaa
and the second form of generating
relation reads
\bea
\omega_-=(d^w\Psi ^{1}\wedge d^w\Psi ^{2}+
d^w\Psi ^{3}\wedge d^w\Psi ^{4}
+\dots+d^w\Psi ^{2N-3}\wedge d^w\Psi ^{2N-2})_-=0.
\label{GenN2}
\eea
Generating relation (\ref{GenN2}) imlies that
$$
\omega=\omega_{ij}(\mathbf{x},\mathbf{y})dw^i\wedge dw^j,
$$
where coefficients
$\omega_{ij}$ are independent of $\lambda$, and
equations (\ref{2formpropN}) for this two-form.

Dressing scheme (\ref{dbar})
requires an obvious modification
\bea
\begin{aligned}
&\dbar\Psi ^{2k-1} =
W_{,2k}(\lambda,\bar\lambda,\Psi ^{1},\dots,\Psi^{2N-2};
\phi^{1},\dots,\phi^{2N}),
\quad W_{,2k}:=\frac{\p W}{\p \Psi^{2k}},
\\
&\dbar\Psi ^{2k} =
W_{,2k-1}(\lambda,\bar\lambda,\Psi ^{1},\dots,\Psi^{2N-2};
\phi^{1},\dots,\phi^{2N}),
\quad W_{,2k-1}:=\frac{\p W}{\p \Psi^{2k-1}},
\end{aligned}
\label{dbarN}
\eea
where $1\leqslant k\leqslant N-1$. Calculation of special solutions of the type (\ref{Theta}) is
completely analogous. We 
start from the dressing data
\begin{multline*}
\frac{1}{2\pi\mathrm{i}}W=
\sum_{i=1}^{M}\delta(\lambda-\mu_i)
F_i(\Psi^1, \Psi^3, \dots, \Psi^{2N-3};
\phi^1,\dots,\phi^{2N})
\\
+
\sum_{i=1}^{M}\delta(\lambda-\nu_i)
G_i(\Psi^2,\Psi^4,\dots,\Psi^{2N-2};\phi^1,\dots,\phi^{2N}),
\end{multline*}
and solutions to the $\dbar$ problem (\ref{dbarN})
are of the form
\bea
\Psi^{2k-1}=q^k+\sum a^k_i y_i+
\sum_{i=1}^{M} \frac{f_i^k}{\lambda-\nu_i}, \quad  
\Psi^2=p^k+\sum b^k_i y_i+
\sum_{i=1}^{M} \frac{g_i^k}{\lambda-\mu_i},
\label{psiN}
\eea
where we take into account vacuum terms.
The functions $f_i^k$, $g_i^k$ are defined as implicit functions by the relations
\be
\begin{aligned}
f_i^k(\mathbf{x},\mathbf{y})&=G_{i,2k}
(\Psi^2,\Psi^4,\dots,\Psi^{2N-2};
\phi^1,\dots,\phi^{2N})\bigr|_{\lambda=\nu_i},
\\
g_i^k(\mathbf{x},\mathbf{y})&=-F_{i,2k-1}
(\Psi^1, \Psi^3, \dots, \Psi^{2N-3};
\phi^1,\dots,\phi^{2N})\bigr|_{\lambda=\mu_i},
\end{aligned}
\label{implN}
\ee
where equations (\ref{implN}) represent a closed
system of $2M(N-1)$  equations for $2M(N-1)$ functions
$f_i^k$, $g_i^k$.
The potential $\Theta$ contains a sum of vacuum 
and regular terms,
\bea
\Theta=\Theta_0+ \wt\Theta,
\quad \Theta_0=\tfrac{1}{2}\sum_{i,k}(a_i^k b_j^k-a_j^k
b_i^k)x^iy^j +
\sum_{i,k} (a_i^k p^k x^i - b_i^k q^k x^i).
\label{THETA}
\eea 
The regular term is defined by multidimensional 
extension of general formula (\ref{HEtau}),
\begin{multline}
\wt\Theta(\mathbf{x},\mathbf{y})=
\\
=
\iint_{G}\frac{d\lambda\wedge d\bar \lambda}{2\pi\mathrm{i}}
\Bigl(\sum_k\wt \Psi^{2k} \dbar \wt \Psi^{2k-1}
-
W(\lambda,\bar \lambda,\Psi^1,\dots,
\Psi^{2N-2};\phi^{1},\dots,\phi^{2N})
\Bigr) 
,
\label{HEtauN}
\end{multline}
and extension of formula (\ref{Theta}) reads
\begin{multline} 
\wt\Theta(\mathbf{x},\mathbf{y},\mathbf{p},\mathbf{q})=
\sum_{i=1}^{M} 
F_i
(\Psi^1, \Psi^3, \dots, \Psi^{2N-3};
\phi^1,\dots,\phi^{2N})
\bigr|_{\lambda=\mu_i}
\\
+
\sum_{i=1}^{M} 
G_i
(\Psi^2,\Psi^4,\dots,\Psi^{2N-2};
\phi^1,\dots,\phi^{2N})
\bigr|_{\lambda=\nu_i}
+\sum_{k=1}^{N-1}
\sum_{i=1}^{M}\sum_{j=1}^{M} \frac{f_i^k g_j^k}{\nu_i - \mu_j},
\label{ThetaN}
\end{multline}
the potential depends on the set of arnitrary functions
of $3N-1$ variables.

The travelling wave reduction $\p_{y^i}=\lambda_i\p_{x^i}$
for some pair of variables $x^i$, $y^i$ corresponds to 
a special dependence of the $\dbar$ data on the function
$\phi^i$, when it enters the data only in the combination
with functions $\Psi^k$, namely
$\Psi ^{2k-1}+\sum \frac{a_i^k\phi^i}{\lambda-\lambda_i}$,
$\Psi ^{2k}+\sum \frac{b_i^k\phi^i}{\lambda-\lambda_i}$
(compare (\ref{dwave}), (\ref{swave})). The
travelling wave reduction for all pairs $x^i$, $y^i$ 
leads to the solution of $2N$-dimensional extension
of general heavenly equation \cite{LVB15}.

The analogues of six-dimensional heavenly equations
are now $2+2N$-dimensional,
\begin{equation}
\Theta_{x^i y^j} - \Theta_{x^j y^i} = \{ \Theta_{x^i}, \Theta_{x^j}\}_{(\mathbf{q,p})} , 
\label{6DHN}
\end{equation}
with the Lax pair
\beaa
\begin{aligned}
D_i \Psi +\{\Theta_{x^i},\Psi\}_{(\mathbf{q,p})}&=0,
\\
D_j \Psi +\{\Theta_{x^j},\Psi\}_{(\mathbf{q,p})}&=0,
\end{aligned}
\eeaa
where the Poisson bracket is         
$$
\{ f_1, f_2\}_{(\mathbf{q,p})} := 
\sum_{k=1}^{N-1}\frac{\partial f_1}{\partial q^k} 
\frac{\partial f_2}{\partial p^k}  -  \frac{\partial f_1}{\partial p^k} 
\frac{\partial f_2}{\partial q^k}. 
$$
Potentials given by expressions (\ref{THETA}),
(\ref{ThetaN}) provide special solutions to 
equations (\ref{6DHN}).

Generaly, $2N+2N$-dimensional generalisation of integrable
structures connected with equation (\ref{TED0})
is similar to generalisation of the second heavenly equation
hierarchy to hyper-K\"ahler hierarchy
\cite{Takasaki89}.
\subsection*{Acknowledgements}
The work of LVB was performed in the framework of
State assignment topic 0033-2019-0006 (Integrable
systems of mathematical physics).

\end{document}